%% file: vcppop.tex
\documentclass[a4paper,UKenglish,cleveref, autoref, thm-restate]{lipics-v2021}

\title{Strengthened Partial-Ordering Based ILP Models for the Vertex Coloring Problem} %

\author{Adalat Jabrayilov}{
  Institute  of  Computer  Science,  University  of  Bonn,
  {Bonn, Germany}}
  {adalat.jabrayilov@uni-bonn.de}
  {https://orcid.org/0000-0002-1098-6358}{}
\author{Petra Mutzel}{
  Institute  of  Computer  Science,  University  of  Bonn,
  {Bonn, Germany}}
  {petra.mutzel@cs.uni-bonn.de}
  {https://orcid.org/0000-0001-7621-971X}{} 
\authorrunning{A.\ Jabrayilov and P.\ Mutzel}

\Copyright{Adalat Jabrayilov and Petra Mutzel}



\category{} %

\funding{This work is funded by the DFG under Germany's Excellence Strategy EXC-2047/1 -- 390685813.} %

\nolinenumbers %

\EventEditors{John Q. Open and Joan R. Access}
\EventNoEds{2}
\EventLongTitle{42nd Conference on Very Important Topics (CVIT 2016)}
\EventShortTitle{CVIT 2016}
\EventAcronym{CVIT}
\EventYear{2016}
\EventDate{December 24--27, 2016}
\EventLocation{Little Whinging, United Kingdom}
\EventLogo{}
\SeriesVolume{42}
\ArticleNo{23}

\usepackage{graphicx}
\usepackage{amsmath} 
\usepackage{amssymb} 
\usepackage{todonotes} %
\usepackage[algo2e,ruled,vlined]{algorithm2e}
\usepackage{setspace}
\usepackage{numprint}
\usepackage{amsfonts} 
\usepackage{tikz} 
\usepackage{hyperref} 
\usepackage[utf8]{inputenc} %
\usepackage{longtable}
\usepackage{ltablex}
\usepackage{nicefrac}

\begin{document} %

\maketitle

\def\subjectto{\textrm{s.t.}}
\def\Z{\mathbb{Z}}
\def\p{\mathcal{P}}
\def\Q{\mathcal{Q}}
\def\x{\mathcal{X}}
\def\X{\mathcal{X}}
\def\I{\mathcal I}
\def\np{\mathcal{NP}}
\def\W{\mathcal{W}}
\def\V{\mathbb{V}}
\def\pop{\textrm{POP}\xspace}
\def\pops{\pop^*}
\def\poph{\textrm{POPH}}

\def\head{\text{head}}
\def\tail{\text{tail}}

\def\ie{i.e.}
\def\eg{e.g.}
\def\Wlog{W.l.o.g.}
\def\len#1{\ell_{#1}}
\def\inarc{\delta^{-}}
\def\red#1{\textcolor{red}{#1}}

\def\todos{0}
\def\todoI#1{
    \ifnum\todos=1 %
      \todo[inline]{#1}
    \fi
  }
\def\vs{\vspace{1cm}}
\def\hs{\hspace{1cm}}

\def\walk{\text {W}}
\def\alert#1{\textcolor{red}{#1}}

\def\ppop{\mathcal{P}}
\def\ppops{\mathcal{P}^*}
\def\pass{\mathcal{A}}

\def\val{\nu}
\def\val#1{\nu(#1)}
\def\gap#1{gap_{#1}}

\def\asssym{\textrm{ASS-S}}
\def\ass{\textrm{ASS}}
\def\rep{\textrm{REP}}
\def\pophybrid{\textrm{POP2}}
\def\minimize{\mathrm{min:}}
\def\subjectto{\mathrm{s.t.}}
\def\todoRM#1{\todo[inline]{#1}}
\def\todoADD#1{\todo[inline,color=green]{#1}}
\def\red#1{\textcolor{red}{#1}}
\def\green#1{\textcolor{green!50!black}{#1}}
\def\changeto#1#2{\\\red{-- #1} \\\green{++ #2}}

\def\msum{\textstyle\sum}
\def\Msum{\displaystyle\sum}

\newdimen\mydim
\newdimen\mydimII
\mydim=1.4cm
\mydimII=0.5\mydim

\def\vcp{\textrm{VCP}\xspace}

\def\vcppop{\textrm{POP}\xspace}
\def\vcppopI{\textrm{POP1}\xspace}
\def\vcppopII{\textrm{POP2}\xspace}
\def\vcppoph{\textrm{POPH}\xspace}
\def\vcppophI{\textrm{POPH1}\xspace}
\def\vcppophII{\textrm{POPH2}\xspace}

\def\vcpass{\textrm{ASS}\xspace}
\def\vcprep{\textrm{REP}\xspace}
\def\vcparep{\textrm{AREP}\xspace}

\def\po{partial-ordering\xspace}
\def\Po{Partial-ordering\xspace}

\def\Forall{\text{for all }}
\def\subjectto{\textrm{s.t.}}
\def\SubjectTo{\textrm{subject to}}

\definecolor{blueI}{RGB}{0,100,255}
\def\TODO#1{\textcolor{blueI}{#1}\marginpar{\textcolor{blueI}{$\longleftarrow$}}}
\def\green#1{\textcolor{green!50!black}{#1}\marginpar{\textcolor{green!50!black}{$\longleftarrow$}}}

\def\a{x}   %

\begin{abstract}
The vertex coloring problem asks for the minimum number of colors that can be assigned to the vertices of a given graph such that each two adjacent vertices get different colors. 
For this NP-hard problem, a variety of integer linear programming (ILP) models have been suggested.
Among them, the \emph{assignment} based and the \emph{\po} based ILP models are attractive due to their simplicity and easy extendability. 
Moreover, on sparse graphs, both models turned out to be among the strongest exact approaches for solving the vertex coloring problem. 
In this work, we suggest additional strengthening constraints for the \emph{\po} based ILP model,
and show that they lead to improved lower bounds of the corresponding LP relaxation.
Our computational experiments confirm that the new constraints are also leading to practical improvements.
In particular, we are able to find the optimal solution of a previously open instance from the DIMACS benchmark set for vertex coloring problems. %
\end{abstract}

\section{Introduction}
\label{sec:intro}

Given a graph $G=(V,E)$, the vertex coloring problem searches for the assignment of colors to the vertices, such that no two adjacent vertices get the same color and the number of colors is minimized. 
The minimum number of colors needed is called the \emph{chromatic number} of $G$. 
Finding the chromatic number is known to be NP-hard \cite{Garey1979}. 
Due to the wide range of applications like register allocation, scheduling, frequency assignment and timetabling problems, there is a vast amount of literature on this problem. 
For surveys, see for example \cite{Husfeldt2015,Lewis2016,Lima2018ExactAF,Malaguti2010}. 

The exact algorithms for the \vcp can be divided into
dynamic programming approaches~\cite{byskov2002chromatic,Christofides1971,Epp03,Lawler1976ANO},
branch-and-bound based enumeration~\cite{Brlaz1979NewMT,Brown1972ChromaticSA,CornailGraham1973,McDiarmid1979DeterminingTC,Segundo2012,Sewell1993AnIA,zykov1949some}, 
and integer linear programming approaches~\cite{Campelo2008,Campos2015,Hansen2009,jabrayilov2018new,Malaguti2011,MehTri96,MenZab06,MenZab08}. 
The dynamic programming algorithms require exponential space of $O(2^{|V|})$ and are of theoretical interest only. 
In contrast, the branch-and-bound algorithms require only polynomial space and are applicable to graphs with up to 80 vertices \cite{Segundo2012}. 
However, the most efficient algorithms for large instances are ILP-based algorithms. 
State-of-the-art ILP models are the \emph{assignment} based model \cite{MenZab06,MenZab08},
the \emph{partial-ordering} based model \cite{jabrayilov2018new}, 
the \emph{representatives} model \cite{Campelo2004,Campelo2008,Campos2015}, 
the set covering and set partitioning models \cite{MehTri96,Malaguti2011,Hansen2009},
and the \emph{ordered binary decision diagrams~(OBDD)} based model \cite{VCPOBDD2022}.

These models can be divided into two groups: \emph{explicit coloring models} and \emph{independent set based models}. %
\emph{Explicit coloring models} directly create a coloring function $c\colon V \rightarrow \mathbb N$. 
Prominent members of this class are the \emph{assignment} and the \emph{\po based models}~\cite{MenZab06,jabrayilov2018new}.
Both models have $O(|V| H)$ variables and $O(|E| H)$ constraints, where $H$ is an upper bound for the chromatic number of the graph~\cite{MenZab06,jabrayilov2018new}. 
If the graph is sparse, then $H$ is expected to be a small number. 
Hence these two models are best suited for sparse graphs.

\emph{Independent set based models} partition the vertex set of the given graph so that every part induces an independent set and thus corresponds to a color class. The models include the representatives model~\cite{Campelo2004}, the set covering and set partitioning models~\cite{MehTri96,Malaguti2011,Hansen2009}, and the \emph{reduced ordered binary decision diagram} approach.
The representatives model~\cite{Campelo2004} has $O(|V| + |\bar E|)$ %
variables and $O(|V| + |V||E|)$ constraints, where $\bar E$ is the set of edges of the complement graph. As the density increases, the number $|\bar E|$ and hence the number of variables decreases. Obviously, this model is particularly suitable for dense graphs, which was also confirmed in practical experiments (see, e.g., \cite{jabrayilov2018new}).
The set covering and set partitioning formulations involve a variable for each independent set and are best suited for dense graphs. 
Unlike the first three models, these models have an exponential number of variables and therefore need to be solved with sophisticated branch-and-price algorithms.
In contrast to the previous models, the decision diagrams approach first transforms the vertex coloring problem for graph $G$ into a reduced ordered binary decision diagram (OBDD), and then computes a constrained flow on this OBDD by solving a network flow based ILP formulation. %
The reduced OBDD is a \emph{directed acyclic graph (DAG)} with exactly one source $s$ and one sink $t$ and contains a directed $(s,t)$-path for each independent set in $G$.
This model is also particularly suitable for dense graphs, since the number of independent sets, and thus the number of $(s,t)$-paths in the OBDD, becomes smaller with increasing density.

A typical solution process of such an ILP model first solves the LP relaxation, i.e., 
the corresponding linear program (LP)  in which the integer requirements are dropped. %
This leads to a first lower bound for the optimal solution value of the integer linear program.
If the LP values of the variables are integral, we have found an integer solution which coincides with the lower bound.
Otherwise, we start a branch-and-bound possibly combined with a cutting plane approach and/or a column generation (pricing) approach.
In the course of the branch-and-bound, the gap between the lower and the upper bounds decreases.
If the lower and upper bounds coincide, we have found an optimal solution for our ILP.

Theoretical and practical analyses have shown that the lower bounds achieved by the assignment 
and the \po based approaches are weaker (i.e., lower) than
those of the models based on independent sets such as the representatives model, the set-covering model, and the OBDD based model.
However, for sparse graphs $G$, the chromatic number $\chi(G)$ and also the gap between the lower bound and $\chi(G)$ 
is generally small.
This observation jointly with the fact that the number of variables and constraints of the assignment and \po ILP models is much smaller than those of the independent set based models, justifies the efficiency of the aforementioned models on sparse graphs.

\medskip

\noindent\textit{Our contribution.}
In this paper, we consider the \po based ILP model. This has recently attracted interest by researchers who applied the \po model to problems related to the \vcp such as the  
\emph{Filter Partitioning Minimization Problem}~\cite{Rahmani2020IntegerLP}, 
the \emph{Multi-Page Labeling Problem}~\cite{gedicke2021}, 
and the \emph{Equitable Chromatic Number Problem}~\cite{olariu2021improving}, 
in which the task is to color the vertices so that adjacent vertices get different colors and the size of the color classes is approximately the same.

\begin{itemize}
\item We propose to add some \emph{strengthening constraints} to the \po model and prove that their addition improves the lower bound of the corresponding LP relaxation. 
\item Our polyhedral investigations show that the strenghtened \po model leads to stronger lower bounds of the LP relaxations in comparison to the assignment model.
\item We present computational experiments that show that our strengthened \po model also leads to practical improvements.
The computational results also include an open DIMACS~\cite{DIMACS:VCP} instance, %
which is solved for the first time (up to our knowledge). %
\item Our experimental comparison with state-of-the-art approaches on the DIMACS benchmark set shows that, as expected, the assignment and the \po based approaches outperform the set-covering and OBDD based models on sparse graphs.
\end{itemize}

\noindent\textit{Outline.}
\autoref{sec:notations} provides the basic notations. %
(\autoref{sec:notations}). 
We describe the assignment and the partial-ordering based ILP models as well as their extensions in \autoref{sec:description:asspop}.
The theoretical analysis of the corresponding polytopes and lower bounds is given in  \autoref{sec:theoreticalresults}.
\autoref{sec:evaluation} describes and discusses the computational experiments.

\section{Notations} 
\label{sec:notations}
We denote the set $\{1,\dots,k\}$ by $[k]$ .
For a graph $G=(V,E)$, we denote its vertex set by $V(G)$ and its edge set by $E(G)$.
For clarity, we may write the edge $e=\{u,v\}$  for $u,v\in V$ as $e=uv$.
By $N(v):=\{u\in V \mid uv \in E\}$ we denote the set of neighbors of $v$.
A \emph{walk} $W$ in a graph is a sequence $W=v_0,e_1,v_1,\dots,v_{k-1},e_{k},v_k$ of not necessarily distinct vertices and edges, such that $e_i=v_{i-1}v_i$ for $i \in [k]$. %
A walk is a \emph{(simple) path} if all its vertices are distinct. 
We may describe $W$ as the sequence of its vertices $v_0,\ldots,v_{k}$ or edges $e_1,\ldots,e_k$.
 A walk $W$ is called \emph{cycle} if $v_0=v_k$.
A cycle $C:=e_1,\dots,e_k$ is called \emph{odd} if $k$ is odd.

\section{Assignment and partial-ordering based models} \label{sec:description:asspop}
In the following we assume that the graph $G=(V,E)$ is connected and has at least one vertex.
Obviously, %
the problem can be solved separately on each connected component.

\subsection{Assignment-based ILP models}
The idea of the assignment based model for the vertex coloring problem is very natural:
namely, the model assigns color $i$ to vertex $v\in V$. 
For this, it introduces \emph{assignment variables} $x_{vi}$ for each vertex $v\in V$ and color $i \in [H]$, with 
$x_{vi}=1$ if vertex $v$ is assigned to color $i$ and $x_{vi}=0$ otherwise. 
$H$ is an upper bound of the chromatic number of $G$ (\eg, the solution value of some heuristic for the problem) and is bounded by $|V|$.
For modelling the objective function, additional binary variables $w_i$ for $i\in [H]$ are needed, which get value 1 if and only if color $i$ is used in the coloring. 
The model is given by:

\begin{eqnarray}
  \vcpass: \hs
  & \min \ \msum_{1 \le i \le H} w_{i}   \notag \\ %
  & \SubjectTo\notag \\ 
  &  \msum_{i=1}^H \a_{vi}   =  1	& \Forall v \in V \label{unambiguous}\\
  & \a_{ui} + \a_{vi}   \le  w_i	& \Forall uv \in E,\ i=1,\ldots,H \label{edge:coloring}\\
  & w_i \le \msum_{v\in V} x_{vi}	&   i =1,\ldots, H\label{mz:cp1}\\
  & w_{i} \le w_{i-1}			&   i =2,\ldots, H  \label{mz:cp2} \\
  & \a_{vi}, w_{i}    \in   \{0,1\}	& \Forall v \in V,\ i =1,\ldots, H  \label{binary-ass}\
\end{eqnarray}
The objective function minimizes the number of used colors. Equation (\ref{unambiguous}) ensures that each vertex receives exactly one color. 
For each edge $uv \in E$ %
there is a set of constraints (\ref{edge:coloring}) making sure that adjacent vertices receive different colors.
These constraints at the same time ensure that a color $i$ can be assigned only to one of the vertices if the variable $w_i$ is set to 1.
  The assignment model with constraints (\ref{unambiguous}), (\ref{edge:coloring}), and  (\ref{binary-ass}) has the advantage that it is simple and easy to use. It can be easily extended to generalizations and/or restricted variants of the graph coloring problem. 
A main drawback of the model is that there are $\binom H {\chi}$ possibilities to select $\chi$ from $H$ colors~\cite{Malaguti2010}. This results in many optimal solutions that are symmetric to each other, 
which is a big problem for branch-and-bound approaches.
In order to overcome this symmetry, Mendez-Diaz and Zabala \cite{MenZab06,MenZab08} have suggested
to add the constraints (\ref{mz:cp1}) and (\ref{mz:cp2}).

\subsection{\Po based ILP models}
Instead of directly assigning a color to the vertices, the \po based model deter\-mines the partial order $(V \cup [H], \prec)$ of the union of the vertex set $V$ and the set $[H]$ of ordered colors, where $H$ is an upper bound for the chromatic number $\chi(G)$. 
This model has two binary variables $l_{v,i}$ and $g_{i,v}$ for each vertex $v\in V$ and each color $i\in [H]$, where 
$l_{v,i}=1$ iff $v \prec i$ (\ie, the color of $v$ is smaller than $i$), and $g_{i,v}=1$ iff $i \prec v$ (the color of $v$ is larger than $i$).
These variables have the following connection with the assignment variables.
Vertex $v$ has color~$i$ (\ie, $x_{v,i}=1$) if and only if $i$ and $v$ are not ordered in $\prec$, \ie, we have $v \not\prec i$ and $i \not\prec v$ ($l_{v,i}=g_{i,v}=0$), thus we have:
\begin{eqnarray}
    \a_{v,i} = 1- (l_{v,i}+g_{i,v}) & \Forall v \in V,\ i = 1,\dots,H.
    \label{connection:asspop}
\end{eqnarray}
The \pop model~\cite{jabrayilov2018new} has the following form, where $q$ is an arbitrary vertex chosen from $V$:
\begin{eqnarray}
  \vcppop: 
    & \min\ 1 + \msum_{1 \le i \le H} g_{i,q}    \notag \\
    & \SubjectTo\notag 
    \\ 
    & l_{v,1} = g_{H,v}      = 0  & \Forall v \in V  
    \label{vcp:VcpPopConstrPositionsStartEnd}    
    \\
    & g_{i-1,v} - g_{i,v}  \ge 0  & \Forall\ v \in V,\ i \in 2,\dots,H  
    \label{vcp:VcpPopConstrUniqPositions1}    
    \\
    & g_{i-1,v} + l_{v,i}    = 1  & \Forall\ v \in V,\ i \in 2,\dots,H 
    \label{vcp:VcpPopConstrUniqPositions2}    
    \\
    & (g_{i,u} + l_{u,i}) + (g_{i,v} + l_{v,i}) \ge 1  & \Forall uv \in E,\ i =1,\dots, H
    \label{vcp:VcpPopConstrAdjVertices}    
    \\
    & g_{i,q} - g_{i,v} \ge 0 & \Forall v \in V,\  i =1,\ldots, H    
    \label{vcp:VcpPopConstrHighestVertexQ}     
    \\
    & g_{i,v}, l_{v,i} \in \{0,1\}  & \Forall v \in V,\  i =1, \ldots, H.
\end{eqnarray}

  Constraints 
  (\ref{vcp:VcpPopConstrPositionsStartEnd})--(\ref{vcp:VcpPopConstrUniqPositions2})
   ensure that each vertex receives exactly one color. 
  Every adjacent pair of vertices must receive different colors.
  This is guaranteed by constraints (\ref{vcp:VcpPopConstrAdjVertices}).
  Constraints 
  (\ref{vcp:VcpPopConstrHighestVertexQ})
  enforce that there is no vertex $v \in V$ with $q \prec v$  %
  in the partial order $\prec$, 
  \ie, the chosen vertex $q$ has the largest used color.

\subsubsection{Strengthening the \vcppop model}
Here we provide additional constraints that strengthen the LP relaxations of the models, i.e., lead to improved lower bounds.

Using equations (\ref{vcp:VcpPopConstrPositionsStartEnd}) and (\ref{vcp:VcpPopConstrUniqPositions2}), one can eliminate all the $l$-variables \cite{jabrayilov2018new}.
Setting these equations to (\ref{vcp:VcpPopConstrAdjVertices}) yields:
\begin{align}
    & g_{1,u} + g_{1,v} \ge 1					    &\Forall uv \in E, \label{vcp:edge1} \\ %
    & (g_{i-1,u} - g_{i,u}) + (g_{i-1,v} - g_{i,v}) \le 1	    &\Forall uv \in E,\ i =2,\ldots, H. \label{vcp:edge2} %
\end{align}

In the following, we show how to strengthen constraints (\ref{vcp:edge1}) and (\ref{vcp:edge2}).
Similar to the assignment model's constraints (\ref{edge:coloring}) we can tighten %
the right hand side of (\ref{vcp:edge1}) and (\ref{vcp:edge2}) by replacing the constant by a variable as follows. %
Constraint (\ref{vcp:edge1}) forces either $g_{1,u}=1$ or $g_{1,v}=1$, and from  (\ref{vcp:VcpPopConstrHighestVertexQ}) we get $g_{1,q}=1$. 
Thus we can replace (\ref{vcp:edge1}) by %
\begin{align}
    & g_{1,u} + g_{1,v} \ge 2-g_{1,q}  &\Forall uv \in E. 
    \label{vcp:edge1s}
\end{align}
If the left hand side of (\ref{vcp:edge2}) is equal to 1, then either $g_{i-1,u}=1$ or $g_{i-1,v}=1$. 
Then $g_{i-1,q}=1$ by (\ref{vcp:VcpPopConstrHighestVertexQ}). 
Thus we can tighten (\ref{vcp:edge2}) as follows:
\begin{align}
    (g_{i-1,u} - g_{i,u}) + (g_{i-1,v} - g_{i,v}) \le g_{i-1,q}  
    && \Forall uv \in E,\ i =2,\dots, H. 
    \label{vcp:edge2s} 
\end{align}
We denote the tightened program by $\vcppopI$:

\begin{eqnarray}
    \vcppopI: 
    & \min\ 1 + \msum_{1 \le i \le H} g_{i,q}    \notag \\
    & \SubjectTo\notag \\
    & g_{H,v} = 0			& \Forall v \in V                    \label{vcp:VcpPopStarRrange} \\
    & g_{i-1,v} - g_{i,v} \ge 0		& \Forall v \in V,\  i =2,\ldots, H  \label{vcp:VcpPopStarUnambiguous1}\\
    & g_{1,u} + g_{1,v} \ge 2-g_{1,q}	& \Forall uv \in E                   \label{vcp:VcpPopStarEdge1} \\
    & (g_{i-1,u} - g_{i,u}) + (g_{i-1,v} - g_{i,v}) \le g_{i-1,q}	    
                                        & \Forall uv \in E,\ i =2,\dots, H   \label{vcp:VcpPopStarEdge2} \\
    & g_{i,q} - g_{i,v} \ge 0		& \Forall v \in V,\  i =1,\ldots, H  \label{vcp:VcpPopStarChi}     \\
    & g_{i,v} \in \{0,1\}               & \Forall v \in V,\  i =1, \ldots, H.\label{vcp:VcpPopStarVars}
\end{eqnarray}

Our second model \vcppopII further strengthens the \vcppopI model.
For all $v \in N(q)$, 
$i \prec v$ ($g_{i,v}=1$) implies $i+1 \prec q$ ($g_{i+1,q}=1$). 
This statement can be formulated as follows:
\begin{align}
    & g_{i+1,q} - g_{i,v} \ge 0	    & \Forall v \in N(q),\  i =1,\ldots, H-1.    \label{vcp:VcpPop2chiNeighbor} 
\end{align}
We extend \vcppopI by the constraints (\ref{vcp:VcpPop2chiNeighbor}) and denote the resulted ILP by \vcppopII.

\subsection{A hybrid partial-ordering based model for \vcp}
\label{subsec:hybridMIP}

The second model in~\cite{jabrayilov2018new} is a hybrid of the assignment and the \po based models and uses the observation that with growing density the \vcppop constraint matrix contains more non-zero elements than that of \vcpass, since 
the \vcppop constraints (\ref{vcp:VcpPopConstrAdjVertices}) contain twice as many non-zero coefficients as the corresponding \vcpass constraints (\ref{edge:coloring}).
Hence one can substitute (\ref{vcp:VcpPopConstrAdjVertices}) by equalities (\ref{connection:asspop}) and the following constraints:
\begin{align}
  \a_{u,i} + \a_{v,i}   & \le  1  & \Forall uv \in E,\ i =1,\ldots, H.
  \label{vcp:VcpPopHConstrAdjVertices}
\end{align}

\subsubsection{Strengthening the hybrid POP model}
We can strengthen the hybrid model analogous to the pure model \vcppop.
If the graph has an edge, we have %
$g_{1,q}=1$. %
Hence, in case $i=1$,
we can replace (\ref{vcp:VcpPopHConstrAdjVertices}) by 
\begin{align}
    \a_{u,1} + \a_{v,1}   & \le  g_{1,q}  & \Forall uv \in E.
    \label{vcp:VcpPopHConstrAdjVerticesI1}
\end{align}

If the color of a vertex $v$ is $i$ (\ie, $\a_{v,i}=1$) for $i \ge 2$, then its color is larger than $i-1$ (\ie, $g_{i-1,v}=1$).
It follows that in the case $\a_{u,i}=1$ or $\a_{v,i}=1$, we have $g_{i-1,u}=1$ or $g_{i-1,v}=1$, and thus 
$g_{i-1,q}=1$ by (\ref{vcp:VcpPopConstrHighestVertexQ}). 
Therefore, we can replace (\ref{vcp:VcpPopHConstrAdjVertices}) by 
\begin{align}
    \a_{u,i} + \a_{v,i}   & \le  g_{i-1,q}  & \Forall uv \in E,\ i =2,\ldots, H.
    \label{vcp:VcpPopHConstrAdjVerticesI2}
\end{align}

Using equations (\ref{vcp:VcpPopConstrPositionsStartEnd}) and (\ref{vcp:VcpPopConstrUniqPositions2}), we eliminate all the $l$-variables. 
This elimination transforms the equalities (\ref{connection:asspop}) into the following form:
\begin{align*}
    &\a_{v,1} = 1         - g_{1,v}	&& \Forall v \in V,		    \\ %
    &\a_{v,i} = g_{i-1,v} - g_{i,v}	&& \Forall v \in V,\ i = 2,\dots,H.    %
\end{align*}
Notice that these constraints jointly with constraints $\a_{v,i} \in [0,1]$ imply the inequalities (\ref{vcp:VcpPopConstrUniqPositions1}). 
Therefore, constraints (\ref{vcp:VcpPopConstrUniqPositions1}) are redundant in the hybrid model.
We denote the tightened hybrid \po program by \vcppophI.
\begin{eqnarray}
    \vcppophI: 
    & \min\ 1 + \msum_{1 \le i \le H} g_{i,q}    \notag
    \\
    & \SubjectTo\notag \\
    & g_{H,v} = 0			& \Forall v \in V                   \label{vcp:VcpPopHStarRrange} \\
    & \a_{v,1} = 1         - g_{1,v}	& \Forall v \in V		    \label{vcp:VcpPopHConnectionAssPop1} \\
    & \a_{v,i} = g_{i-1,v} - g_{i,v}	& \Forall v \in V,\ i = 2,\dots,H   \label{vcp:VcpPopHConnectionAssPop2} \\
    & \a_{u,1} + \a_{v,1} \le  g_{1,q}  & \Forall uv \in E                  \label{vcp:VcpPopHStarEdge1} \\
    & \a_{u,i} + \a_{v,i} \le g_{i-1,q} & \Forall uv \in E,\ i =2,\ldots,H  \label{vcp:VcpPopHStarEdge2} \\
    & g_{i,q} - g_{i,v} \ge 0		& \Forall v \in V,\  i =1,\ldots,H  \label{vcp:VcpPopHStarChi}     \\
    & \a_{v,i}, g_{i,v} \in \{0,1\}     & \Forall v \in V,\  i =1,\ldots,H. \label{vcp:VcpPopHStarVars}
\end{eqnarray}
Similar to the pure pop model \vcppopII, we also generate a hybrid model \vcppophII, which extends \vcppophI by the constraints (\ref{vcp:VcpPop2chiNeighbor}).

\begin{remark}[non-zero coefficients] \label{vcp:PopVsPoph}
Simply counting the non-zero coefficients in the constraint matrix of the pure and the hybrid \pop models gives $(6|V| +5 |E|) \cdot H -|V|-2|E|$ and $(9|V| +3|E|)\cdot H$, respectively. 
This means that the constraint matrix of the hybrid model becomes sparser than that of the pure model with increasing $|E|$.
\end{remark}

\section{Polyhedral Results} 
\label{sec:theoreticalresults}
In this section we show the polyhedral advantages of the \po based models over the assignment model for the \vcp.
In particular, we study the LP relaxations of the models POP1 and POP2. 
By solving such an LP relaxation, we obtain a fractional solution.
The solution value of these relaxations is also called the \emph{dual bound} to the optimal solution value of the original ILP model.

We first consider the dual bounds of the \po based models \vcppopI and \vcppopII and show that they are larger than two for all connected graphs with $\chi(G)\ge 3$.

\subsection{Dual Bound of \vcppopI}

\begin{lemma}
  Let $C:=q,v_1,v_2,\ldots,v_{2k},q$ {with $k\ge 1$} be an odd cycle in $G$, where $q$ is the chosen vertex  in the \po based models. Then:
  \begin{eqnarray*}
     \val {\vcppopI} \ge 
	\left\{
	    \begin{array}[2]{ll}
		    2+ \frac13         & k=1\\
		    2 + \frac 1 {k+1}  & k\ge2.
	    \end{array}
	\right.
  \end{eqnarray*}
  \label{vcp:lemmaOddCyclePopI}
\end{lemma}
\begin{proof}
Adding (\ref{vcp:VcpPopStarEdge2}) for an edge $uv$ and each $i=2,\dots,H$ gives: 
\begin{eqnarray*}
  \msum_{i=1}^{H-1} g_{i,q} \ge g_{1,u} + g_{1,v} - g_{H,u} - g_{H,v}	&& \Forall uv \in E. 
\end{eqnarray*}
As $g_{H,q} = g_{H,u} = g_{H,v}= 0$ by (\ref{vcp:VcpPopStarRrange}), we can rewrite these inequalities as follows:
\begin{eqnarray}
  \msum_{i=1}^{H} 
  g_{i,q} \ge g_{1,u} + g_{1,v}	    &&\Forall uv \in E. 
  \label{vcp:lemmaOddCyclePopIx1}
\end{eqnarray}

\begin{description}
  \item[Case $k=1$:] $C$ has exactly three vertices, \ie, $C:=q,v_1,v_2,q$.
  Adding the inequalities (\ref{vcp:lemmaOddCyclePopIx1}) for edges $qv_1,qv_2 \in E$ gives:

  \begin{eqnarray}
    2\msum_{i=1}^{H} g_{i,q} \ge 2g_{1,q} + g_{1,v_1} + g_{1,v_2}.
    \label{vcp:lemmaOddCyclePopIx2}
  \end{eqnarray}

  The sum of the constraints 
  \begin{align*}
    \msum_{i=1}^{H} g_{i,q}  &\ge g_{1,q} + \frac12 g_{1,v_1} + \frac12 g_{1,v_2}    && //\ \frac12 \times (\ref{vcp:lemmaOddCyclePopIx2}) 
    \\
                          0  &\ge \frac12(2-g_{1,q} -g_{1,v_1} -g_{1,v_2})           && //\ \frac12 \times (\ref{vcp:VcpPopStarEdge1}) \mbox{ for edge $v_1v_2\in E$}
    \\
                          0  &\ge \frac16(2-g_{1,q} -g_{1,q} -g_{1,v_1})             && //\ \frac16 \times (\ref{vcp:VcpPopStarEdge1}) \mbox{ for edge $qv_1\in E$}
    \\
                          0  &\ge \frac16 (-g_{1,q} +g_{1,v_1})                      && //\ \frac16 \times (\ref{vcp:VcpPopStarChi}) \mbox{ for $v_1\in V$}
  \end{align*}

  imply $\sum_{i=1}^{H} g_{i,q} \ge 1\frac 13$. It follows:
  \begin{align*}
     \val\vcppopI = 1+\msum_{i=1}^{H} g_{i,q} & \ge 2\frac 13.
  \end{align*}

\medskip
  \item[Case $k\ge2$:]
Adding (\ref{vcp:VcpPopStarEdge1}) for edges $v_{2j-1}v_{2j}$ with $j=1,\dots,k$ yields:
\begin{align}
      \msum_{j=1}^{k} ( g_{1,v_{2j-1}} + g_{1,v_{2j}} ) \ge k(2 - g_{1,q}).
      \label{vcp:lemmaOddCyclePopIx3}
\end{align}
Adding (\ref{vcp:lemmaOddCyclePopIx1}) for each edge $v_{2j}v_{2j+1}$ with $j=1,\dots,k-1$ yields:

\begin{eqnarray}
  \msum_{j=1}^{k-1} 
  \msum_{i=1}^{H} g_{i,q} 
  \ge
  \msum_{j=1}^{k-1} ( g_{1,v_{2j}} + g_{1,v_{2j+1}} ).
  \label{vcp:lemmaOddCyclePopIx4}
\end{eqnarray}

Adding (\ref{vcp:lemmaOddCyclePopIx3}) and (\ref{vcp:lemmaOddCyclePopIx4}) leads to: 

\begin{eqnarray}
  &
  \msum_{j=1}^{k-1} 
  \msum_{i=1}^{H} g_{i,q} 
  & \ge 
    \msum_{j=1}^{k-1} ( g_{1,v_{2j}} + g_{1,v_{2j+1}} )
     - \msum_{j=1}^{k} ( g_{1,v_{2j-1}} + g_{1,v_{2j}} ) 
     + k(2-g_{1,q})  \notag
  \\
  && \displaystyle
   = k(2-g_{1,q}) - g_{1, v_1}- g_{1, v_{2k}}.
  \label{vcp:lemmaOddCyclePopIx5}
\end{eqnarray}

The sum of the constraints 

\begin{align*}
  \msum_{j=1}^{k-1} \msum_{i=1}^{H} g_{i,q} &\ge k(2-g_{1,q}) - g_{1, v_1}- g_{1, v_{2k}} && //\ (\ref{vcp:lemmaOddCyclePopIx5})
  \\
  \msum_{i=1}^{H} g_{i,q}                  &\ge g_{1,q} + g_{1,v_1}                        && //\ (\ref{vcp:lemmaOddCyclePopIx1}) \mbox{ for edge } qv_1\in E 
  \\
  \msum_{i=1}^{H} g_{i,q}                  &\ge g_{1,q} + g_{1,v_{2k}}                        && //\ (\ref{vcp:lemmaOddCyclePopIx1}) \mbox{ for edge } qv_{2k}\in E 
\end{align*}
gives:
\begin{align*}
    (k+1) \msum_{i=1}^{H} g_{i,q} 
      & 
      \ge k(2-g_{1,q}) + 2 g_{1,q}
      \ge 2k - g_{1,q}(k-2)
      \\
      & \ge 2k - (k-2)    & //\quad g_{1,q} \le 1
      \\
      & = k +2.
\end{align*}

It follows $\sum_{i=1}^{H} g_{i,q} \ge \frac{k+2}{k+1}$.
Thus we have: %
\begin{align*}
  \val\vcppopI = 1 + \msum_{i=1}^{H} g_{i,q} 
    & \ge 1+ \frac {k+2} {k +1}
    = 2 + \frac 1 {k +1}.
\end{align*}
\end{description}
\end{proof}
\begin{theorem} \label{vcp:TheoremVcpPopILPbound}
  Let $G$ be a graph with $\chi(G)>2$. Then the improved model \vcppopI satisfies: %
  \begin{align*}
    \val\vcppopI \ge 2 + \frac 1 {|V|} >2.
  \end{align*}
\end{theorem}
\begin{proof}
From $\chi(G)>2$ it
follows that $G$ contains an odd cycle $C$ with $2k+1$ vertices for some $k \ge 1$. 
Since $G$ is connected, there is a shortest simple path $P$ that starts at $q$ and ends at a vertex $w$ of $C$ such that $V(P) \cap V(C)=\{w\}$.
Then the path $q,\dots,w$, the cycle $C$, and the path $w,\dots,q$ build an odd cycle with at most 
\begin{align*}
  (|V|-2k-1) + (2k+1) + (|V|-2k-1) =2(|V|-k-1)+1
\end{align*}
vertices. Let $k':=(|V|-k-1)$. By Lemma \ref{lemma:oddcycle} we have:
  \begin{align*}
    \val\vcppopI \ge 2 + \frac 1 {k'+2} = 2 + \frac 1 {(|V|-k-1)+2} = 2 + \frac 1 {|V|-k+1} \ge 2 + \frac 1 {|V|},
  \end{align*}
where the last inequality follows from $k\ge1$. %
\end{proof}

\subsection{Dual Bound of \vcppopII}

\begin{lemma}
  Let $C:=q,v_1,v_2,\ldots,v_{2k},q$ with $k\ge 1$ be an odd cycle in $G$, where $q$ is the chosen vertex  in the \po based models. Then:
   \begin{align*}
       \val\vcppopII \ge 2 + \frac 1 {k+1}.
   \end{align*}
  \label{lemma:oddcycle}
\end{lemma}
\begin{proof}
Recall that \vcppopII contains all constraints of \vcppopI. 
Then for $k\ge2$ we have 
\begin{align*}
   \val\vcppopII \ge \val\vcppopI \ge 2 + \frac 1 {k+1}.
\end{align*}
Thus the claim holds for $k\ge2$. 
So it is sufficient to consider the case $k=1$.
 Then $C$ has three vertices, \ie, $C:=q,v_1,v_2,q$.
 Adding (\ref{vcp:VcpPop2chiNeighbor}) for a neighbor $v \in N(q)$ and each $i=1,\dots,H-1$ yields:
 \begin{eqnarray*}
   \msum_{i=2}^{H} g_{i,q}   \ge	
   \msum_{i=1}^{H-1} g_{1,v}    && \Forall v \in N(q).
 \end{eqnarray*}
 Adding $g_{1,q}$ on both sides of these inequalities yields:
 \begin{eqnarray}
   \msum_{i=1}^{H} g_{i,q}   \ge	g_{1,q} + 
   \msum_{i=1}^{H-1} g_{1,v}	    
   && \Forall v \in N(q).
   \label{vcp:lemmaOddCyclePopIx6}
 \end{eqnarray}
 The sum of constraints (\ref{vcp:lemmaOddCyclePopIx6}) for both neighbors $v_1$ and $v_2$ of $q$ gives:
 \begin{eqnarray}
   2 \msum_{i=1}^{H} g_{i,q}   \ge    2g_{1,q} + 
     \msum_{i=1}^{H-1} g_{i,v_1} + 
     \msum_{i=1}^{H-1} g_{i,v_2}.
   \label{vcp:lemmaOddCyclePopIx7}
 \end{eqnarray}
 The sum of the constraints 
 \begin{align*}
   \msum_{i=1}^{H} g_{i,q}   &\ge g_{1,q} + 
   \frac12 \msum_{i=1}^{H-1} (g_{i,v_1} + g_{i,v_2}) && //\ \frac12 \times (\ref{vcp:lemmaOddCyclePopIx7}) 
   \\
 		        0  &\ge 2-g_{1,q} -g_{1,v_1} -g_{1,v_2}           && //\ 1 \times (\ref{vcp:VcpPopStarEdge1}) \mbox{ for $v_1v_2\in E$}
   \\
 		        0 &\ge \frac12 \big( -g_{1,q} + (g_{1,v_1} - g_{2,v_1}) + (g_{1,v_2} - g_{2,v_2}) \big)
 			&& //\ \frac12 \times (\ref{vcp:VcpPopStarEdge2}) \mbox{ for $v_1v_2\in E$, $i=2$}
 \end{align*}
 imply that: %
 \begin{align*}
   \msum_{i=1}^{H} g_{i,q}   &\ge 2 - \frac12 g_{1,q} + \frac12 \msum_{i=3}^{H-1} (g_{i,v_1} + g_{i,v_2}) 
   \\
                            &\ge 2 - \frac12 g_{1,q} 
   \\
                            &\ge 2 - \frac12	    & //\quad g_{1,q} \le 1
   \\
                            &\ge 1 \frac12.
 \end{align*}
 It follows for $k=1$:
 \begin{align*}
    \val\vcppopII = 1+\msum_{i=1}^{H} g_{i,q} & \ge 2\frac 12 =2 +\frac 1{k+1}.
 \end{align*}
\end{proof}

\subsection{Strength-relationship}
Next, we consider the strength-relationship between the assignment and  the \po based models.  
For polyhedral comparisons we use the term \emph{strength of an LP relaxation}~\cite{Polzin2009}. %
We denote by $\val M$ the value of the LP relaxation of the ILP $M$.
Let $M$ and $M'$ be two ILPs for a minimization problem~$P$.
We say $M$ is stronger than $M'$, if $\val M \ge \val {M'}$ holds for all instances of $P$; 
$M$ is strictly stronger than $M'$ if $M$ is stronger than $M'$, but $M'$ is not stronger than $M$. %

\begin{lemma}
 \vcppopII is strictly stronger than \vcppopI.
  \label{vcp:LemmaPop2Pop1}
\end{lemma}

\begin{proof}
The LP relaxation of  \vcppopII is stronger than that of $\vcppopI$, %
since the former contains all constraints of the latter.
To show that \vcppopI is not stronger than \vcppopII, it is sufficient to construct a \vcp instance with strict inequality $\val\vcppopII > \val\vcppopI$.
Let $G$ be a complete graph with four vertices $q,a,b,c$, and let $H=4$.
Consider the following feasible solution for the linear program $\vcppopI$:
\begin{align}
(g_{1,q}, g_{2,q}, g_{3,q}, g_{4,q}) &= 
(\nicefrac 45, \nicefrac 25, \nicefrac 15, 0)
\notag
\\
(g_{1,v},g_{2,v},g_{3,v},g_{4,v}) &= 
(\nicefrac 35, \nicefrac 15, 0, 0)
&& \forall v \in V \setminus \{q\}. \notag 
\end{align}
of value $1 + \sum_{i=1}^{H} g_{i,q} =2.4$. 
It follows that $\val\vcppopI \le 2.4$. 
Since $G$ contains the 3-cycle $q,a,b,q$ we have $\val\vcppopII \ge 2.5$ by Lemma \ref{lemma:oddcycle}, so $\val\vcppopII > \val\vcppopI$, as claimed.
\end{proof}

\begin{lemma}
 \vcppopI is strictly stronger than \vcpass.
 \label{vcp:LemmaPop1Ass}
\end{lemma}

\begin{proof}
We first show that \vcppopI is stronger than \vcpass.
If $G$ consists of only one vertex, then both models have the same value 1, so suppose $G$ has some edges.
The model \vcpass has the following feasible solution of value $2$, valid for any graph~\cite{Malaguti2010}:
\begin{align*}
& \Forall v \in V: &&  x_{vi}= \left\{
				    \begin{array}[2]{ll}
					    \nicefrac12 & i \in \{1,2\} \\
					     0	    & i>2, 
				    \end{array}
				    \right.
\\
& \Forall i \in [H]: &&  w_{i}= \left\{
				    \begin{array}[2]{ll}
					    1 & i \in \{1,2\} \\
					    0   & i>2. 
				    \end{array}
				    \right.
\end{align*}
It follows that the objective value of the assignment model, independent of the graph, is at most~2, \ie, $\val\vcpass \le 2$.
On the other hand, if $G$ has an edge, %
the objective value of \vcppopI is at least~2, \ie, $\val\vcppopI \ge 2$. 
The reason is the following: 
Adding (\ref{vcp:VcpPopStarEdge1}) and (\ref{vcp:lemmaOddCyclePopIx1}) for any edge $uv \in E$ gives:
$  \msum_{i=1}^{H} g_{i,q}  \ge 2 -  g_{1,q} \ge 1.  $
Therefore we have:
\begin{eqnarray*}
\val\vcppopI = 1 + \msum_{i=1}^{H} g_{i,q} \ge 2.
\end{eqnarray*}
Therefore, for any graph we have $\val\vcppopI \ge \val\vcpass$.
On the other hand for any graph $G$ with $\chi(G) > 2$, we have $\val\vcppopI>2$ by Theorem \ref{vcp:TheoremVcpPopILPbound},
while $\val\vcpass=2$.
\end{proof}

\paragraph*{Chosen vertex $\boldsymbol q$  in model \vcpass}
In contrast to the assignment model, \vcppopI has constraints, namely (\ref{vcp:VcpPopStarChi}), which force the chosen vertex $q$ to take the largest used color. 
The natural question is how the strength of the assignment model changes if we add similar constraints to it.
Equalities (\ref{connection:asspop}), (\ref{vcp:VcpPopConstrPositionsStartEnd}), and (\ref{vcp:VcpPopConstrUniqPositions2}) give
\begin{eqnarray}
      g_{1,v} =& 1- x_{v,1} & \Forall v \in V, \\
      g_{i,v} =& g_{i-1,v} - x_{v,i} & \Forall v \in V,\ i = 2,\dots,H,
\end{eqnarray}
which jointly imply
\begin{eqnarray}
      g_{i,v} = 1- \msum_{j=1}^i x_{v,j} & \Forall v \in V,\ i = 1,\dots,H.
\end{eqnarray}
So we can add (\ref{vcp:VcpPopStarChi}) to the assignment model as follows:
\begin{eqnarray}
    1-\msum_{j=1}^i x_{q,j} \ge 1-\msum_{j=1}^{i} x_{v,j} & \Forall v \in V\setminus \{q\},\ i = 1,\dots,H.
    \label{vcp:VcpAssStarChi}
\end{eqnarray}
Notice that these constraints cannot cut off the solution $(x,w)$ shown as an example in the proof of Lemma \ref{vcp:LemmaPop1Ass}.
So, strengthening \vcpass via (\ref{vcp:VcpAssStarChi}) does not make it stronger than \vcppopI, since the fact $\val\vcpass \le 2$ remains unchanged.

\begin{lemma}
 \vcpass is strictly stronger than \vcppop.
  \label{vcp:LemmaAssPop}
\end{lemma}

\begin{proof}
In case $G$ consists of only one vertex, both models have the same objective value~1. 
Therefore we consider the case that $G$ has some edges.
The model \vcppop has the solution 
\begin{align*}
\Forall v \in V: 
&&  l_{vi}= \left\{
				    \begin{array}[2]{ll}
					    0       & i = 1 \\
					    \nicefrac12 & i = 2 \\
					     1	    & i>2,
				    \end{array}
				    \right.
&&  g_{iv}= \left\{
				    \begin{array}[2]{ll}
					    \nicefrac12 & i = 1 \\
					     0	    & i \ge 2, 
				    \end{array}
				    \right.
\end{align*}
of objective value 
$1+\sum_{i=1}^{H} g_{i,q} = 1 \frac 12$, independent of the graph, so
$\val\vcppop \le 1\frac12$.
Adding (\ref{edge:coloring}) for an arbitrary edge $ab \in E$ and $1\le i \le H$ gives 
$$\msum_{i=1}^H w_i \ge \sum_{i=1}^H x_{ai} + \sum_{i=1}^H x_{bi}.$$
From (\ref{unambiguous}) it follows that $\sum_{i=1}^H x_{ai}=\sum_{i=1}^H x_{bi}=1$, which implies that
$\sum_{i=1}^H w_{i}\ge2$.  
Therefore, we have $\val\vcpass \ge 2$, which is strictly larger than 
$\val\vcppop \le 1\frac12$.
So \vcpass is strictly stronger than \vcppop.

\end{proof}

\begin{lemma} \label{vcp:LemmaPopPoph}
 The corresponding pure and hybrid \pop models are equally strong, \ie:
\begin{eqnarray}
  & \vcppop   &= \vcppoph, \\
  & \vcppopI  &= \vcppophI, \\
  & \vcppopII &= \vcppophII.
\end{eqnarray}
\end{lemma}

\begin{proof}

Using
equality $\a_{v,i} = 1- (l_{v,i}+g_{i,v})$ (constraint (\ref{connection:asspop}))
in inequality $\a_{u,i} + \a_{v,i} \le  1$ (constraint (\ref{vcp:VcpPopHConstrAdjVertices})) of \vcppoph, transforms it to \vcppop,
thus \vcppop = \vcppoph.
Using equality $\a_{v,1}=1-g_{1,v}$ (constraint (\ref{vcp:VcpPopHConnectionAssPop1})) in inequality $\a_{u,1}+\a_{v,1}\le g_{1,q}$ (constraint (\ref{vcp:VcpPopHConstrAdjVerticesI1})) and 
setting the equality $\a_{v,i}=g_{i-1,v}-g_{i,v}$ (constraint~(\ref{vcp:VcpPopHConnectionAssPop2})) in inequality $\a_{u,i}+\a_{v,i}\le g_{i-1,q}$ (constraint~(\ref{vcp:VcpPopHConstrAdjVerticesI2})) 
gives (\ref{vcp:edge1s}) and (\ref{vcp:edge2s}), respectively.
In this way, one can transform 
\vcppophI to \vcppopI and
\vcppophII to \vcppopII.
\end{proof}

The following theorem summarizes the polyhedral results, shown in Lemmas 
\ref{vcp:LemmaPop2Pop1}--\ref{vcp:LemmaPopPoph}.
\\

\begin{theorem}
The following strength-relationships hold: 
$$
\begin{array}{llll}
    \vcppopII & \succ \vcppopI  & \succ \vcpass  & \succ \vcppop \\
    \ \ \|  &  \ \ \ \ \  \|    & & \ \ \ \ \  \|     \\ 
    \vcppophII & \succ \vcppophI &		 & \succ \vcppoph. 
\end{array}
$$
  \label{theorem:strength}
\end{theorem}

\section{Computational Results} 
\label{sec:evaluation}

\def\MalagutiEtAl{Malaguti, Monaci, and Toth}
\def\HeldEtAl{Held, Cook, and Sewell}
\def\Hoeve{van~Hoeve}

\def\nt#1{\normalsize{\textbf{#1}}}
\def\Mour{\textrm{M}_{\textrm{our}}}
\def\Mmmt{\textrm{M}_{\textrm{MMT}}}
\def\Mobdd{\textrm{M}_{\textrm{vH}}}
\def\Bour{\mathcal{D}_{\textrm{our}}}
\def\Bmmt{\mathcal{D}_{\textrm{MMT}}}
\def\Bobdd{\mathcal{D}_{\textrm{vH}}}
\def\Bdimacs{\mathcal{D}}
As previously analyzed, we expect that the explicit ILP models for the vertex coloring problem such as the assignment and \po based models are best suited for sparse graphs, whereas the models relying on independent sets are beneficial for dense graphs.
In our computational experiments we are interested in answering the following questions:

\begin{description}
\item [(Q1):] How do the strengthened \po based models compare to other explicit coloring models such as the assignment model?

\item [(Q2):] 
Does one of the strengthened pure and hybrid partial-ordering based models dominate the other?

\item [(Q3):] How do the strengthened models compare to the sophisticated state-of-the-art algorithms relying on independent sets? %

\end{description}

\subsection{Implementation}
\label{sse:preprocessing}

We have implemented the assignment and the strongest pure and hybrid \po based models \vcppopII and \vcppophII using the Gurobi-python~API.
Our implementation involves the following preprocessings, which are widely used in the literature \cite{gualandi2012exact,Hansen2009,jabrayilov2018new,Malaguti2011,Malaguti2010,MenZab06,MenZab08}:

\begin{enumerate}
  \item [(a)]
A vertex $u$ is \emph{dominated} by vertex $v$, $v\not= u$, if the neighborhood of $u$ is a subset of the neighborhood of $v$. In this case, the vertex $u$ can be deleted, the remaining graph can be colored, and at the end $u$ can get the color of $v$.

  \item [(b)] To reduce the number of variables we are interested in getting a small upper bound $H$ for the number of needed colors.

  \item [(c)]
    One can fix more variables in the assignment and \pop models when precoloring a clique $Q$ with
    $\max ( |Q|H + |\delta(Q)| )$, 
    where $\delta(Q):=\{(u,v) \in E \colon |\{u,v\} \cap Q| = 1\}$ (see \cite{jabrayilov2018new}).

  \item [(d)] 
  Since any clique represents a valid lower bound for the \vcp, in the case of equal lower and upper bounds 
    the optimal value has been found, hence no ILP needs to be solved. 
\end{enumerate}

The preprocessing is done 
using heuristics from the python library \texttt{networkx}. %
The special vertex $q$ is chosen arbitrarily from the clique $Q$ found in the preprocessing, and the remaining vertices in $Q$ are precolored with colors $1,\dots,|Q|-1$.
  Moreover, our implementations involve the branching priority option of the Gurobi solver;
  the higher the degree of a vertex, the higher the branching priority of the associated variables.

\subsection{Test setup and benchmark set of graphs}

To solve the implemented models, we used Gurobi 6.5 single-threadedly
on the machine with Intel E5-2640@2.60GHz %
running Ubuntu 18.04. %
For computational comparisons we also compiled the code of the \HeldEtAl~\cite{Held2012} model on a same machine and with the same Gurobi version.
The \texttt{dfmax} benchmark \cite{DIMACS:dfmax}, which is used in the DIMACS 
challenge to compare the results 
obtained with different machines, needs 5.0s for \texttt{r500.5} on our machine.
We performed our evaluations on a subset of the DIMACS \cite{DIMACS:VCP} benchmark set.
From 137~DIMACS graphs we have chosen all instances that have at most \numprint{100000} edges. 
The largest one, namely \texttt{4-FullIns\_5}, of the chosen 116 instances has 
\numprint{4146} vertices and \numprint{77305} edges.
We have set a time limit of three hours.

\subsection{Experimental evaluation}

\begin{table}[t]
\scalebox{0.70}{
  \begin{tabular}{l | ll ll | ll}
    Comparison             &\vcppopII                    &\vcppophII                        &\vcpass                      &HCS\cite{Held2012} &MMT\cite{Malaguti2011}   &van Hoeve\cite{VCPOBDD2022} \\ 
    \hline
    \nt{DIMACS subset (size)}   &$\Bour$ (116)                &$\Bour$ (116)                     &$\Bour$ (116)                &$\Bour$ (116)        &$\Bmmt$ (115)         &$\Bobdd$ (137) \\
    \nt{Machine (CPU)}                &$\Mour$(2.6GHz)	      &$\Mour$(2.6GHz)			 &$\Mour$(2.6GHz)              &$\Mour$(2.6GHz)     &$\Mmmt$(2.4GHz)      &$\Mobdd$(2.33GHz) \\
    \nt{\#solved}               &74/116 (3h)		      &\nt{84/116} (3h)		         &81/116 (3h)                  &62/116 (3h)         &70/115 (10h) &50/137 (1h) \\
    \nt{\#solved/$\Bour\cap \Bmmt$}     &66/103 (3h)		      &\nt{74/103} (3h)		         &71/103 (3h)                  &50/103 (3h) &68/103 (10h) & \\
    \nt{\#solved/$\Bour \cap \Bobdd$(1h)} &74/116 ($\frac{2.33}{2.6}$h) &\nt{82/116} ($\frac{2.33}{2.6}$h) &80/116 ($\frac{2.33}{2.6}$h) &62/116 ($\frac{2.33}{2.6}$h) &	          &49/116 (1h) \\
    \nt{average runtime}        &\nt{71.1} sec                         &261.2 sec                             &168.7 sec                        &92.1 sec                 &434.2 sec    &74.7 sec \\
    \hline
  \end{tabular}
}
\caption{Experimental results for state-of-the-art ILP models for the DIMACS \cite{DIMACS:VCP} benchmark set $\Bdimacs=\Bobdd$ and subsets $\Bour$ and $\Bmmt$
} %
\label{vcp:tableVcpDimacsSummary}
\end{table}

Table \ref{vcp:tableVcpDimacsSummary} summarizes some experimental results on the DIMACS instances.
Column 1 shows the subjects of the comparison, while the remaining column 2--7 display the associated results of the models 
\vcppopII,
\vcppophII,
\vcpass,
\HeldEtAl~\cite{Held2012}, referred to as HCS, 
\MalagutiEtAl~\cite{Malaguti2011}, referred to as MMT,   
and \Hoeve~\cite{VCPOBDD2022}. 
While Columns 2--5 show the results of our computations, the last two columns are taken from~\cite{Malaguti2011,VCPOBDD2022}. 
Let $\Bdimacs$ denote the DIMACS~\cite{DIMACS:VCP} benchmark set and let 
$\Bour$, $\Bmmt$ and $\Bobdd$ denote the subsets of $\Bdimacs$ used in our experiments, 
in \MalagutiEtAl~\cite{Malaguti2011} and in \Hoeve~\cite{VCPOBDD2022}, respectively.
$\Bour$ and $\Bmmt$ contain 116 and 115 instances of $\Bdimacs$, respectively, while $\Bobdd$ encloses the whole $\Bdimacs$.
The intersection of $\Bour \cap \Bmmt$ contains 103 common instances.
As $\Bobdd = \Bdimacs$, we have $\Bour \cap \Bobdd = \Bour$.

The first two rows show the instances and machines used in the experiments.
The next three rows display the number of solved instances by the algorithms.
The items in these rows have the form ''$n/N\ (t)$``, which shows the number $n$ of solved instances out of $N$ instances within the time limit $t$.
Row ``\#solved`` shows the total number of solved instances by these six models.
We can see that the model \vcppophII has solved the largest number of instances, namely 84 out of 116 instances within the time limit of 3 hours, to provable optimality.
\begin{figure}[bt]
  \def\scale{1.1}
  \begin{center}
  \includegraphics[page=1,scale=\scale]{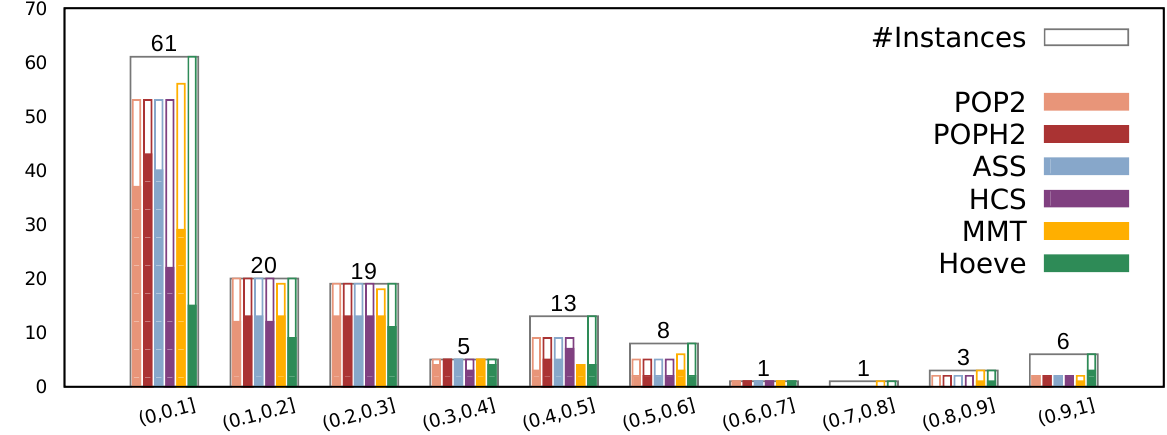}
  \end{center}
  \caption{Number of solved DIMACS instances depending on graph density $\frac{2|E|}{|V|(|V|-1)}$
 (The density-groups %
 have 
 61, 20, 19, 5, 13, 8, 1, 1, 3, 6 instances, respectively, and 137 in total.) 
  }
  \label{vcp:PlotsDIMACSresolution}
\end{figure}

Figure \ref{vcp:PlotsDIMACSresolution} visualizes some details about this row.
We partitioned the 137 DIMACS instances into ten classes $(0,0.1],\dots,(0.9,1]$ according to their densities $\frac{2|E|}{|V|(|V|-1)}$. 
These classes contain 
61, 20, 19, 5, 13,  8, 1, 1, 3, 6
instances, respectively.
The figure provides the following information for each density class: 
The wide box shows the total number of instances in the class.
For each ILP model, the corresponding filled bar shows the number of instances solved by that model, while the unfilled bar shows the number of instances that were evaluated but could not be solved by the model.
For example in the first class there are 61 instances.
We evaluated the models \vcppopII, \vcppophII, \vcpass and HCS using 53 out of the 61 instances, where the models solved 37, 43, 40, 22, respectively.
\MalagutiEtAl~\cite{Malaguti2011} and \Hoeve~\cite{VCPOBDD2022} 
evaluated 56 and 61 out of the 61 instances and solved 29 and 15 instances, respectively.
We can see that most DIMACS instances are sparse graphs; the first density class, in which \vcppophII dominates all other tested models, includes about half of the instances.

Row ``\#solved/$\Bour \cap \Bmmt$`` of Table~\ref{vcp:tableVcpDimacsSummary} compares, our experimental results with the results reported by \MalagutiEtAl~\cite{Malaguti2011} based on the 103 common DIMACS instances.
Also here we can see that the model \vcppophII has solved the largest number of instances, namely 74.
The last row ``\#solved/$\Bour \cap \Bobdd$(1h)`` compares, based on the 116 common DIMACS instances, our experimental results with the results reported by \Hoeve~\cite{VCPOBDD2022}.
For a fair comparison, it is interesting to see the number of instances solved for each algorithm at similar computing times.       
To this end, this row shows the number of instances solved for each algorithm at the computing time that roughly corresponds to 1 hour of the machine used in \cite{VCPOBDD2022}.
For example the computing time $\frac{2.33}{2.6}$ hours in our machine roughly corresponds to 1 hour of the $\frac{2.33GHz}{2.6GHz}$ times slower machine used in \cite{VCPOBDD2022}.
Also here we can see that the model \vcppophII has solved the largest number of instances, namely 82 out of the 116 instances.
Finally, the last row shows for each model the average running times spent for the solved instances.

\def\dimacsvcp{DIMACS~\cite{DIMACS:VCP}\xspace}

Answering our questions from the beginning, we can state the following:

\begin{description}

\item [(Q1):] 
Our computational study shows that the strengthened hybrid model \vcppophII outperforms the assignment formulation \vcpass on the DIMACS benchmark set, \ie, it solves more \dimacsvcp instances than \vcpass. %
This is not true for the pure model \vcppopII. %

\item [(Q2):] 
Although \vcppopII and \vcppophII are equally strong by Theorem~\ref{theorem:strength}, the latter solves more instances than the former.
The explanation lies in the fact described in Remark~\ref{vcp:PopVsPoph}, which implies that the constraint matrix of the hybrid model becomes sparser than that of the pure model for graphs with $\frac{|E|}{|V|} \ge 1.5$, which is the case for all considered 116 instances.

 \item [(Q3):] 
 Our computational results and the studies reported in the literature (also see Table~\ref{vcp:tableVcpDimacsSummary}) show that the strengthened models \vcppopII and \vcppophII dominate the sophisticated state-of-the-art algorithms \cite{Held2012,Malaguti2011,VCPOBDD2022} on the sparse \dimacsvcp instances with density $\frac{2|E|}{|V|(|V|-1)}\le0.1$ (Figure~\ref{vcp:PlotsDIMACSresolution}).
 Moreover, a closer look at the single instances (Table \ref{vcp:tableVcpDimacs1} in Appendix) shows that the hybrid model \vcppophII is the only model, which can solve all five DIMACS \texttt{GPIA} graphs that originated from the estimation of the sparse Jacobian matrix problem~\cite{ColemanMore1983}.
 To our knowledge, \vcppophII is the first one that can solve the open instance \texttt{abb313GPIA}.

\end{description}

\bibliography{vcppop}

\newpage
\appendix

\section{Computational Results}

Table \ref{vcp:tableVcpDimacs1} shows the detailed results for the 116 \dimacsvcp instances that we have used in our computational evaluations.
Recall that we have considered only the instances with up to \numprint{100000} edges.

Columns 1-3 show the instance names and sizes.
Columns 4--15 display the lower and upper bounds as well as the running times of the models \vcppopII, \vcppophII, \vcpass, and HCS~\cite{Held2012} that we have obtained within a time limit of three hours in our computations. %
The times are provided in seconds. %

Columns 16--18 are taken from the literature~\cite{Malaguti2011} and show the results of sophisticated state-of-the-art algorithm suggested by \MalagutiEtAl~\cite{Malaguti2011}. 
These results have been obtained within a time limit of ten hours.
The columns contain the character ``-'' if the associated instances were not considered in~\cite{Malaguti2011}.

The last three columns are taken from~\cite{VCPOBDD2022} and show the results of the sophisticated algorithm suggested by \Hoeve~\cite{VCPOBDD2022}. 
Note that these results have been obtained within a time limit of one hour.

\begin{sidewaystable}
\def\nt#1{\normalsize{\textbf{#1}}}
\small
\scalebox{0.7}{
  \setlength{\tabcolsep}{7pt} 
  \setlength{\extrarowheight}{2pt}
  \begin{tabular}{l  rr || ccc | ccc | ccc | ccc  || ccc | ccc }
               &    &     &    &\vcppopII  &         &   &\vcppophII   &         &    &\vcpass    &         &        &HCS~\cite{Held2012}     &               &        &MMT~\cite{Malaguti2011}  &         &    &v.Hoeve\cite{VCPOBDD2022}    &       \\
  Instance     &V   &E    &lb  &ub         &time[s]  &lb  &ub          &time[s]  &lb      &ub     &time[s]  &lb      &ub                      &time[s]        &lb      &ub                       &time[s]  &lb       &ub                      &time[s]       \\

    \hline
    \input{tableDimacsGurobi651_p1.tex}
    \hline
  \end{tabular}
}
\caption{Results of the six models for the \dimacsvcp instances with up to \numprint{100000} edges}
\label{vcp:tableVcpDimacs1}
\end{sidewaystable}

\begin{sidewaystable}
\def\nt#1{\normalsize{\textbf{#1}}}
\small
\scalebox{0.7}{
  \setlength{\tabcolsep}{7pt} 
  \setlength{\extrarowheight}{2pt}
  \begin{tabular}{l  rr || ccc | ccc | ccc | ccc  || ccc | ccc }
               &    &     &    &\vcppopII  &         &   &\vcppophII   &         &    &\vcpass    &         &        &HCS~\cite{Held2012}     &               &        &MMT~\cite{Malaguti2011}  &         &    &v.Hoeve\cite{VCPOBDD2022}    &       \\
  Instance     &V   &E    &lb  &ub         &time[s]  &lb  &ub          &time[s]  &lb      &ub     &time[s]  &lb      &ub                      &time[s]        &lb      &ub                       &time[s]  &lb       &ub                      &time[s]       \\

    \hline
\input{tableDimacsGurobi651_p2.tex}
    \hline
  \end{tabular}
}
\caption*{Table \ref{vcp:tableVcpDimacs1} (continued)} 
\end{sidewaystable}

\begin{sidewaystable}
\def\nt#1{\normalsize{\textbf{#1}}}
\small
\scalebox{0.7}{
  \setlength{\tabcolsep}{7pt} 
  \setlength{\extrarowheight}{2pt}
  \begin{tabular}{l  rr || ccc | ccc | ccc | ccc  || ccc | ccc }
               &    &     &    &\vcppopII  &         &   &\vcppophII   &         &    &\vcpass    &         &        &HCS~\cite{Held2012}     &               &        &MMT~\cite{Malaguti2011}  &         &    &v.Hoeve\cite{VCPOBDD2022}    &       \\
  Instance     &V   &E    &lb  &ub         &time[s]  &lb  &ub          &time[s]  &lb      &ub     &time[s]  &lb      &ub                      &time[s]        &lb      &ub                       &time[s]  &lb       &ub                      &time[s]       \\

    \hline
\input{tableDimacsGurobi651_p3.tex} 
    \hline
  \end{tabular}
}
\caption*{Table \ref{vcp:tableVcpDimacs1} (continued)} 
\end{sidewaystable}

\end{document}

%% file: tableDimacsGurobi651_p1.tex
1-FullIns\_3     &30    &100    &4.0     &4.0        &0.0      &4.0      &4.0        &0.0      &4.0        &4.0        &0.0      &4          &4    &0.0      &-       &-    &-        &4        &4    &0.06     \\
1-FullIns\_4     &93    &593    &5.0     &5.0        &0.0      &5.0      &5.0        &0.0      &5.0        &5.0        &0.0      &4          &5    &timeout  &4       &5    &timeout  &4        &5    &timeout  \\
1-FullIns\_5     &282   &3247   &6.0     &6.0        &0.7      &6.0      &6.0        &0.9      &6.0        &6.0        &0.7      &4          &6    &timeout  &4       &6    &timeout  &4        &6    &timeout  \\
1-Insertions\_4  &67    &232    &5.0     &5.0        &57.8     &5.0      &5.0        &33.3     &5.0        &5.0        &8.7      &4          &5    &timeout  &3       &5    &timeout  &3        &5    &timeout  \\
1-Insertions\_5  &202   &1227   &4.0     &6.0        &timeout  &4.0      &6.0        &timeout  &4.0        &6.0        &timeout  &4          &6    &timeout  &3       &6    &timeout  &3        &6    &timeout  \\
1-Insertions\_6  &607   &6337   &4.0     &7.0        &timeout  &4.0      &7.0        &timeout  &4.0        &7.0        &timeout  &4          &7    &timeout  &3       &7    &timeout  &3        &7    &timeout  \\
2-FullIns\_3     &52    &201    &5.0     &5.0        &0.0      &5.0      &5.0        &0.0      &5.0        &5.0        &0.0      &5          &5    &0.0      &5       &5    &2.9      &5        &5    &0.83     \\
2-FullIns\_4     &212   &1621   &6.0     &6.0        &0.0      &6.0      &6.0        &0.0      &6.0        &6.0        &0.0      &5          &6    &timeout  &5       &6    &timeout  &5        &6    &timeout  \\
2-FullIns\_5     &852   &12201  &7.0     &7.0        &1.8      &7.0      &7.0        &1.5      &7.0        &7.0        &1.8      &5          &7    &timeout  &5       &7    &timeout  &4        &7    &timeout  \\
2-Insertions\_3  &37    &72     &4.0     &4.0        &0.1      &4.0      &4.0        &0.1      &4.0        &4.0        &0.3      &4          &4    &246.1    &-       &-    &-        &3        &4    &timeout  \\
2-Insertions\_4  &149   &541    &4.0     &5.0        &timeout  &4.0      &5.0        &timeout  &4.0        &5.0        &timeout  &3          &5    &timeout  &3       &5    &timeout  &3        &5    &timeout  \\
2-Insertions\_5  &597   &3936   &4.0     &6.0        &timeout  &4.0      &6.0        &timeout  &4.0        &6.0        &timeout  &3          &6    &timeout  &3       &6    &timeout  &3        &6    &timeout  \\
3-FullIns\_3     &80    &346    &6.0     &6.0        &0.0      &6.0      &6.0        &0.0      &6.0        &6.0        &0.0      &6          &6    &0.0      &6       &6    &2.9      &6        &6    &13.15    \\
3-FullIns\_4     &405   &3524   &7.0     &7.0        &0.0      &7.0      &7.0        &0.0      &7.0        &7.0        &0.0      &6          &7    &timeout  &6       &7    &timeout  &5        &7    &timeout  \\
3-FullIns\_5     &2030  &33751  &8.0     &8.0        &3.2      &8.0      &8.0        &3.0      &8.0        &8.0        &2.7      &6          &8    &timeout  &6       &8    &timeout  &5        &8    &timeout  \\
3-Insertions\_3  &56    &110    &4.0     &4.0        &0.9      &4.0      &4.0        &0.7      &4.0        &4.0        &1.9      &3          &4    &timeout  &3       &4    &timeout  &3        &4    &timeout  \\
3-Insertions\_4  &281   &1046   &4.0     &5.0        &timeout  &4.0      &5.0        &timeout  &3.0        &5.0        &timeout  &3          &5    &timeout  &3       &5    &timeout  &3        &5    &timeout  \\
3-Insertions\_5  &1406  &9695   &3.0     &6.0        &timeout  &3.0      &6.0        &timeout  &3.0        &6.0        &timeout  &3          &6    &timeout  &2       &6    &timeout  &3        &6    &timeout  \\
4-FullIns\_3     &114   &541    &7.0     &7.0        &0.0      &7.0      &7.0        &0.0      &7.0        &7.0        &0.0      &7          &7    &0.0      &7       &7    &3.4      &7        &7    &99.58    \\
4-FullIns\_4     &690   &6650   &8.0     &8.0        &0.0      &8.0      &8.0        &0.0      &8.0        &8.0        &0.0      &7          &8    &timeout  &7       &8    &timeout  &7        &8    &timeout  \\
4-FullIns\_5     &4146  &77305  &9.0     &9.0        &7.7      &9.0      &9.0        &9.6      &9.0        &9.0        &11.0     &7          &9    &timeout  &7       &9    &timeout  &7        &9    &timeout  \\
4-Insertions\_3  &79    &156    &4.0     &4.0        &20.3     &4.0      &4.0        &10.0     &4.0        &4.0        &167.2    &3          &4    &timeout  &3       &4    &timeout  &3        &4    &timeout  \\
4-Insertions\_4  &475   &1795   &3.0     &5.0        &timeout  &3.0      &5.0        &timeout  &3.0        &5.0        &timeout  &3          &5    &timeout  &3       &5    &timeout  &3        &5    &timeout  \\
5-FullIns\_3     &154   &792    &8.0     &8.0        &0.0      &8.0      &8.0        &0.0      &8.0        &8.0        &0.0      &8          &8    &0.0      &8       &8    &4.6      &8        &8    &1233.00  \\
5-FullIns\_4     &1085  &11395  &9.0     &9.0        &0.1      &9.0      &9.0        &0.1      &9.0        &9.0        &0.1      &8          &9    &timeout  &8       &9    &timeout  &8        &9    &timeout  \\
\nt{abb313GPIA}       &1555  &53356  &8.0     &10.0       &timeout  &\nt{9.0}      &\nt{9.0}        &\nt{9784.9}   &8.0        &11.0       &timeout  &8          &10   &timeout  &8       &9    &timeout  &8        &10   &timeout  \\
anna             &138   &493    &11      &11         &0        &11       &11         &0        &11         &11         &0        &11         &11   &0.0      &11      &11   &3.6      &11       &11   &0.01     \\
ash331GPIA       &662   &4181   &4.0     &4.0        &15.5     &4.0      &4.0        &154.2    &4.0        &4.0        &846.9    &4          &6    &timeout  &4       &4    &45.9     &4        &5    &timeout  \\
ash608GPIA       &1216  &7844   &4.0     &4.0        &1902.7   &4.0      &4.0        &282.9    &4.0        &4.0        &337.0    &4          &6    &timeout  &4       &4    &2814.8   &3        &5    &timeout  \\
\nt{ash958GPIA}       &1916  &12506  &\nt{4.0}     &\nt{4.0}        &\nt{342.7}    &\nt{4.0}      &\nt{4.0}        &\nt{813.1}    &4.0        &6.0        &timeout  &$-\infty$  &6    &timeout  &3       &4    &timeout  &3        &5    &timeout  \\
david            &87    &406    &11      &11         &0        &11       &11         &0        &11         &11         &0        &11         &11   &0.0      &11      &11   &0.2      &11       &11   &0.01     \\
DSJC1000.1       &1000  &49629  &6.0     &29.0       &timeout  &6.0      &29.0       &timeout  &$-\infty$  &29.0       &timeout  &$-\infty$  &25   &timeout  &5       &20   &timeout  &6        &26   &timeout  \\
DSJC125.1        &125   &736    &5.0     &5.0        &3.1      &5.0      &5.0        &4.0      &5.0        &5.0        &1.3      &5          &6    &timeout  &5       &5    &142.0    &5        &6    &timeout  \\
DSJC125.5        &125   &3891   &11.0    &20.0       &timeout  &13.0     &20.0       &timeout  &13.0       &19.0       &timeout  &17         &22   &timeout  &17      &17   &18050.8  &14       &22   &timeout  \\
DSJC125.9        &125   &6961   &35.0    &47.0       &timeout  &40.0     &44.0       &timeout  &42.0       &44.0       &timeout  &44         &45   &timeout  &44      &44   &3896.9   &44       &44   &25.11    \\
DSJC250.1        &250   &3218   &5.0     &9.0        &timeout  &5.0      &9.0        &timeout  &5.0        &9.0        &timeout  &7          &10   &timeout  &6       &8    &timeout  &5        &10   &timeout  \\
DSJC250.5        &250   &15668  &12.0    &41.0       &timeout  &14.0     &41.0       &timeout  &14.0       &38.0       &timeout  &26         &36   &timeout  &20      &28   &timeout  &16       &36   &timeout  \\
DSJC250.9        &250   &27897  &40.0    &$+\infty$  &timeout  &41.0     &$+\infty$  &timeout  &42.0       &90.0       &timeout  &72         &77   &timeout  &71      &72   &timeout  &71       &74   &timeout  \\
DSJC500.1        &500   &12458  &5.0     &17.0       &timeout  &6.0      &18.0       &timeout  &6.0        &18.0       &timeout  &$-\infty$  &16   &timeout  &5       &12   &timeout  &5        &16   &timeout  \\
DSJC500.5        &500   &62624  &13.0    &70.0       &timeout  &14.0     &$+\infty$  &timeout  &13.0       &$+\infty$  &timeout  &43         &65   &timeout  &16      &48   &timeout  &18       &68   &timeout  \\

%% file: tableDimacsGurobi651_p2.tex
DSJR500.1        &500   &3555   &12.0    &12.0       &2.7      &12.0     &12.0       &0.7      &12.0       &12.0       &0.5      &12         &12   &1165.9   &12      &12   &35.3     &12       &12   &0.03     \\
DSJR500.5        &500   &58862  &115.0   &135.0      &timeout  &120.0    &$+\infty$  &timeout  &115.0      &128.0      &timeout  &122        &132  &timeout  &122     &122  &342.2    &112      &126  &timeout  \\
flat300\_20\_0   &300   &21375  &11.0    &$+\infty$  &timeout  &12.0     &$+\infty$  &timeout  &12.0       &43.0       &timeout  &20         &20   &10.2     &-       &-    &-        &16       &42   &timeout  \\
flat300\_26\_0   &300   &21633  &11.0    &45.0       &timeout  &13.0     &$+\infty$  &timeout  &12.0       &43.0       &timeout  &26         &26   &24.3     &-       &-    &-        &16       &43   &timeout  \\
flat300\_28\_0   &300   &21695  &12.0    &45.0       &timeout  &13.0     &$+\infty$  &timeout  &14.0       &43.0       &timeout  &28         &41   &timeout  &-       &-    &-        &16       &44   &timeout  \\
fpsol2.i.1       &269   &11654  &65      &65         &0        &65       &65         &0        &65         &65         &0        &65         &65   &0.6      &65      &65   &10.6     &65       &65   &10.33    \\
fpsol2.i.2       &363   &8691   &30      &30         &0        &30       &30         &0        &30         &30         &0        &30         &30   &0.4      &30      &30   &11.2     &30       &30   &0.21     \\
fpsol2.i.3       &363   &8688   &30      &30         &0        &30       &30         &0        &30         &30         &0        &30         &30   &0.3      &30      &30   &10.0     &30       &30   &0.25     \\
games120         &120   &638    &9       &9          &0        &9        &9          &0        &9          &9          &0        &9          &9    &0.0      &9       &9    &0.2      &9        &9    &92.74    \\
homer            &556   &1629   &13      &13         &0        &13       &13         &0        &13         &13         &0        &13         &13   &0.1      &-       &-    &-        &10       &13   &timeout  \\
huck             &74    &301    &11      &11         &0        &11       &11         &0        &11         &11         &0        &11         &11   &0.0      &11      &11   &0.2      &11       &11   &0.23     \\
inithx.i.1       &519   &18707  &54      &54         &0        &54       &54         &0        &54         &54         &0        &54         &54   &2.1      &54      &54   &21.0     &54       &54   &3.17     \\
inithx.i.2       &558   &13979  &31      &31         &0        &31       &31         &0        &31         &31         &0        &31         &31   &0.7      &31      &31   &9.2      &31       &31   &121.09   \\
inithx.i.3       &559   &13969  &31      &31         &0        &31       &31         &0        &31         &31         &0        &31         &31   &0.8      &31      &31   &9.9      &31       &31   &163.70   \\
jean             &77    &254    &10      &10         &0        &10       &10         &0        &10         &10         &0        &10         &10   &0.0      &10      &10   &0.2      &10       &10   &0.01     \\
le450\_15a       &450   &8168   &15.0    &16.0       &timeout  &15.0     &15.0       &1032.0   &15.0       &15.0       &1157.2   &15         &17   &timeout  &15      &15   &0.4      &15       &18   &timeout  \\
le450\_15b       &450   &8169   &15.0    &15.0       &1088.2   &15.0     &15.0       &1384.6   &15.0       &15.0       &1261.2   &15         &17   &timeout  &15      &15   &0.2      &15       &17   &timeout  \\
le450\_15c       &450   &16680  &15.0    &23.0       &timeout  &15.0     &$+\infty$  &timeout  &15.0       &27.0       &timeout  &15         &15   &349.7    &15      &15   &3.1      &15       &25   &timeout  \\
le450\_15d       &450   &16750  &15.0    &23.0       &timeout  &15.0     &$+\infty$  &timeout  &15.0       &27.0       &timeout  &15         &15   &383.3    &15      &15   &3.8      &15       &25   &timeout  \\
le450\_25a       &450   &8260   &25      &25         &0        &25       &25         &0        &25         &25         &0        &25         &25   &2.7      &25      &25   &0.1      &25       &25   &0.17     \\
le450\_25b       &450   &8263   &25      &25         &0        &25       &25         &0        &25         &25         &0        &25         &25   &2.5      &25      &25   &0.1      &25       &25   &0.14     \\
le450\_25c       &450   &17343  &25.0    &27.0       &timeout  &25.0     &27.0       &timeout  &25.0       &27.0       &timeout  &25         &28   &timeout  &25      &25   &1356.6   &25       &29   &timeout  \\
le450\_25d       &450   &17425  &25.0    &27.0       &timeout  &25.0     &27.0       &timeout  &25.0       &28.0       &timeout  &25         &29   &timeout  &25      &25   &66.6     &25       &28   &timeout  \\
le450\_5a        &450   &5714   &5.0     &9.0        &timeout  &5.0      &5.0        &103.9    &5.0        &5.0        &44.7     &5          &6    &timeout  &5       &5    &0.3      &5        &10   &timeout  \\
le450\_5b        &450   &5734   &5.0     &7.0        &timeout  &5.0      &5.0        &89.5     &5.0        &5.0        &214.6    &5          &5    &760.8    &5       &5    &0.2      &5        &9    &timeout  \\
le450\_5c        &450   &9803   &5.0     &5.0        &622.4    &5.0      &5.0        &24.3     &5.0        &5.0        &23.5     &5          &5    &368.6    &5       &5    &0.1      &5        &8    &timeout  \\
le450\_5d        &450   &9757   &5.0     &5.0        &760.9    &5.0      &5.0        &24.6     &5.0        &5.0        &29.3     &5          &6    &timeout  &5       &5    &0.2      &5        &5    &0.03     \\
miles1000        &128   &3216   &42.0    &42.0       &6.7      &42.0     &42.0       &1.2      &42.0       &42.0       &0.5      &42         &42   &0.1      &42      &42   &0.2      &42       &42   &1.87     \\
miles1500        &128   &5198   &73      &73         &0        &73       &73         &0        &73         &73         &0        &73         &73   &0.2      &73      &73   &0.1      &73       &73   &1.10     \\
miles250         &125   &387    &8       &8          &0        &8        &8          &0        &8          &8          &0        &8          &8    &0.0      &8       &8    &5.0      &8        &8    &0.01     \\
miles500         &128   &1170   &20      &20         &0        &20       &20         &0        &20         &20         &0        &20         &20   &0.0      &20      &20   &3.7      &20       &20   &0.04     \\
miles750         &128   &2113   &31.0    &31.0       &1.0      &31.0     &31.0       &0.6      &31.0       &31.0       &0.3      &31         &31   &0.1      &31      &31   &0.2      &31       &31   &0.12     \\
mug100\_1        &100   &166    &4.0     &4.0        &0.1      &4.0      &4.0        &0.1      &4.0        &4.0        &0.1      &4          &4    &1.0      &4       &4    &14.4     &3        &4    &timeout  \\
mug100\_25       &100   &166    &4.0     &4.0        &0.2      &4.0      &4.0        &0.2      &4.0        &4.0        &0.2      &4          &4    &0.9      &4       &4    &12.0     &3        &4    &timeout  \\
mug88\_1         &88    &146    &4.0     &4.0        &0.1      &4.0      &4.0        &0.1      &4.0        &4.0        &0.1      &4          &4    &0.6      &4       &4    &9.6      &3        &4    &timeout  \\
mug88\_25        &88    &146    &4.0     &4.0        &0.3      &4.0      &4.0        &0.1      &4.0        &4.0        &0.1      &4          &4    &0.6      &4       &4    &10.6     &3        &4    &timeout  \\
mulsol.i.1       &138   &3925   &49      &49         &0        &49       &49         &0        &49         &49         &0        &49         &49   &0.2      &49      &49   &0.2      &49       &49   &0.40     \\
mulsol.i.2       &173   &3885   &31      &31         &0        &31       &31         &0        &31         &31         &0        &31         &31   &0.1      &31      &31   &4.7      &31       &31   &0.19     \\
mulsol.i.3       &174   &3916   &31      &31         &0        &31       &31         &0        &31         &31         &0        &31         &31   &0.1      &31      &31   &0.2      &31       &31   &0.19     \\
mulsol.i.4       &175   &3946   &31      &31         &0        &31       &31         &0        &31         &31         &0        &31         &31   &0.1      &31      &31   &0.2      &31       &31   &0.19     \\

%% file: tableDimacsGurobi651_p3.tex
mulsol.i.5       &176   &3973   &31      &31         &0        &31       &31         &0        &31         &31         &0        &31         &31   &0.1      &31      &31   &6.0      &31       &31   &0.19     \\
myciel3          &11    &20     &4.0     &4.0        &0.0      &4.0      &4.0        &0.0      &4.0        &4.0        &0.0      &4          &4    &0.0      &4       &4    &0        &4        &4    &0.04     \\
myciel4          &23    &71     &5.0     &5.0        &0.1      &5.0      &5.0        &0.1      &5.0        &5.0        &0.1      &5          &5    &8.1      &5       &5    &4        &5        &5    &7.03     \\
myciel5          &47    &236    &6.0     &6.0        &5.2      &6.0      &6.0        &2.5      &6.0        &6.0        &3.5      &5          &6    &timeout  &$-\infty$       &$+\infty$    &timeout  &5        &6    &timeout  \\
myciel6          &95    &755    &6.0     &7.0        &timeout  &7.0      &7.0        &2286.0   &7.0        &7.0        &7730.0   &5          &7    &timeout  &4       &7    &timeout  &4        &7    &timeout  \\
myciel7          &191   &2360   &5.0     &8.0        &timeout  &5.0      &8.0        &timeout  &6.0        &8.0        &timeout  &5          &8    &timeout  &5       &8    &timeout  &4        &8    &timeout  \\
qg.order30       &900   &26100  &30.0    &$+\infty$  &timeout  &30.0     &30.0       &70.1     &30.0       &30.0       &55.9     &30         &32   &timeout  &30      &30   &0.2      &30       &32   &timeout  \\
qg.order40       &1600  &62400  &40.0    &55.0       &timeout  &40.0     &40.0       &4318.9   &$-\infty$  &64.0       &timeout  &40         &42   &timeout  &40      &40   &2.9      &40       &45   &timeout  \\
queen10\_10      &100   &1470   &10.0    &12.0       &timeout  &10.0     &11.0       &timeout  &10.0       &11.0       &timeout  &11         &14   &timeout  &11      &11   &686.9    &10       &14   &timeout  \\
queen11\_11      &121   &1980   &11.0    &13.0       &timeout  &11.0     &13.0       &timeout  &11.0       &13.0       &timeout  &11         &11   &61.9     &11      &11   &1865.7   &10       &15   &timeout  \\
queen12\_12      &144   &2596   &12.0    &14.0       &timeout  &12.0     &14.0       &timeout  &12.0       &14.0       &timeout  &12         &16   &timeout  &12      &13   &timeout  &10       &16   &timeout  \\
queen13\_13      &169   &3328   &13.0    &15.0       &timeout  &13.0     &15.0       &timeout  &13.0       &15.0       &timeout  &13         &17   &timeout  &13      &14   &timeout  &10       &17   &timeout  \\
queen14\_14      &196   &4186   &14.0    &17.0       &timeout  &14.0     &17.0       &timeout  &14.0       &16.0       &timeout  &14         &19   &timeout  &14      &15   &timeout  &10       &19   &timeout  \\
queen15\_15      &225   &5180   &15.0    &18.0       &timeout  &15.0     &18.0       &timeout  &15.0       &18.0       &timeout  &15         &20   &timeout  &15      &16   &timeout  &10       &21   &timeout  \\
queen16\_16      &256   &6320   &16.0    &19.0       &timeout  &16.0     &19.0       &timeout  &16.0       &19.0       &timeout  &16         &21   &timeout  &16      &17   &timeout  &10       &22   &timeout  \\
queen5\_5        &25    &160    &5.0     &5.0        &0.0      &5.0      &5.0        &0.0      &5.0        &5.0        &0.0      &5          &5    &0.0      &5       &5    &0.2      &5        &5    &0.01     \\
queen6\_6        &36    &290    &7.0     &7.0        &0.4      &7.0      &7.0        &0.2      &7.0        &7.0        &0.2      &7          &7    &5.0      &7       &7    &3        &7        &7    &2.35     \\
queen7\_7        &49    &476    &7.0     &7.0        &4.4      &7.0      &7.0        &0.3      &7.0        &7.0        &0.3      &7          &7    &0.0      &7       &7    &0.2      &7        &7    &2.77     \\
queen8\_12       &96    &1368   &12.0    &12.0       &8.1      &12.0     &12.0       &1.9      &12.0       &12.0       &1.1      &12         &12   &9.6      &12      &12   &0.2      &9        &12   &timeout  \\
queen8\_8        &64    &728    &9.0     &9.0        &195.7    &9.0      &9.0        &14.2     &9.0        &9.0        &38.9     &9          &11   &timeout  &9       &9    &3.6      &9        &9    &310.31   \\
queen9\_9        &81    &1056   &9.0     &10.0       &timeout  &10.0     &10.0       &1369.9   &10.0       &10.0       &1621.2   &10         &12   &timeout  &10      &10   &36.6     &10       &12   &timeout  \\
r1000.1          &1000  &14378  &20.0    &20.0       &126.3    &20.0     &20.0       &19.5     &20.0       &20.0       &17.9     &20         &20   &1.2      &-       &-    &-        &20       &20   &0.20     \\
r125.1c          &125   &7501   &46      &46         &0        &46       &46         &0        &46         &46         &0        &46         &46   &0.0      &-       &-    &-        &46       &46   &2.08     \\
r125.1           &122   &209    &5       &5          &0        &5        &5          &0        &5          &5          &0        &5          &5    &0.0      &-       &-    &-        &5        &5    &0.02     \\
r125.5           &125   &3838   &34.0    &36.0       &timeout  &36.0     &36.0       &3.2      &36.0       &36.0       &3.2      &36         &36   &72.9     &-       &-    &-        &36       &36   &212.90   \\
r250.1c          &250   &30227  &64      &64         &0        &64       &64         &0        &64         &64         &0        &64         &64   &0.2      &-       &-    &-        &64       &64   &31.19    \\
r250.1           &250   &867    &8       &8          &0        &8        &8          &0        &8          &8          &0        &8          &8    &0.0      &-       &-    &-        &8        &8    &7.32     \\
r250.5           &250   &14849  &65.0    &66.0       &timeout  &65.0     &65.0       &17.7     &65.0       &65.0       &7.3      &65         &65   &2046.6   &-       &-    &-        &65       &67   &timeout  \\
school1          &385   &19095  &14.0    &14.0       &36.9     &14.0     &14.0       &39.6     &14.0       &14.0       &36.2     &14         &14   &88.4     &14      &14   &0.4      &14       &14   &1.67     \\
school1\_nsh     &352   &14612  &14.0    &14.0       &35.7     &14.0     &14.0       &36.1     &14.0       &14.0       &35.2     &14         &14   &82.1     &14      &14   &17.0     &14       &14   &15.59    \\
wap05a           &905   &43081  &50      &50         &0        &50       &50         &0        &50         &50         &0        &50         &50   &7.9      &50      &50   &293.2    &50       &50   &15.80    \\
wap06a           &947   &43571  &40.0    &$+\infty$  &timeout  &40.0     &$+\infty$  &timeout  &40.0       &$+\infty$  &timeout  &40         &44   &timeout  &40      &40   &175.0    &40       &43   &timeout  \\
will199GPIA      &701   &6772   &7.0     &7.0        &9.0      &7.0      &7.0        &2.5      &7.0        &7.0        &5.5      &7          &7    &2.3      &7       &7    &80.7     &6        &7    &timeout  \\
zeroin.i.1       &126   &4100   &49      &49         &0        &49       &49         &0        &49         &49         &0        &49         &49   &0.2      &49      &49   &3.8      &49       &49   &0.32     \\
zeroin.i.2       &157   &3541   &30      &30         &0        &30       &30         &0        &30         &30         &0        &30         &30   &0.1      &30      &30   &4.4      &30       &30   &0.17     \\
zeroin.i.3       &157   &3540   &30      &30         &0        &30       &30         &0        &30         &30         &0        &30         &30   &0.0      &30      &30   &4.5      &30       &30   &0.20     \\